\documentclass[a4paper,11pt]{article}
\usepackage[utf8]{inputenc}
\usepackage{amsmath,amssymb,color}
\usepackage[usenames,dvipsnames]{xcolor}
\usepackage{amsthm,enumerate}
\usepackage{float}
\newtheorem{lemma}{Lemma}
\newtheorem{theorem}{Theorem}
\newtheorem{remark}{Remark}
\newtheorem{definition}{Definition}
\DeclareMathOperator*{\argmin}{arg\,min}
\DeclareMathOperator*{\argmax}{arg\,max}
\usepackage{graphicx}
\usepackage{fullpage,natbib}

\newcommand*\Dif{\mathop{}\!\mathrm{D}}

\usepackage{verbatim} 

\title{Confidence Intervals for Maximin Effects in Inhomogeneous Large-Scale
Data}
\author{Dominik Rothenh\"ausler, Nicolai Meinshausen, Peter B\"uhlmann \\ Seminar f\"ur Statistik, ETH Z\"urich}

\begin{document}
\maketitle 

\begin{abstract}
One challenge of large-scale data analysis is that the assumption of an
identical distribution for all samples is often not realistic. An optimal linear regression might, for example, be markedly different for distinct groups of the data.
Maximin effects have been proposed as a computationally attractive way to
estimate effects that are common across all data without fitting a mixture distribution explicitly. 
So far just point estimators of the common maximin effects have been
proposed in \citet{maximin}.  Here we propose asymptotically valid confidence regions for these effects. 
\end{abstract}

\section{Introduction}
Large-scale regression analysis often has to deal with inhomogeneous
data in the sense that samples are not drawn independently from
the same distribution. The optimal
regression coefficient might for example be markedly different in distinct
groups of the data or vary slowly over a chronological ordering of the
samples. One option is then to
either model the exact variation of the regression vector with a
varying-coefficient model in the latter case
\citep{hastie1993varying,fan1999statistical} or to fit a mixture
distribution in the former
\citep{aitkin1985estimation,mclachlan2004finite,figueiredo2002unsupervised}.
For large-scale analysis with many groups of data samples or many
predictor variables this approach might be too expensive
computationally and also yield more information than necessary in
settings where one is just interested in effects that are present in
all sub-groups of data. A maximin effect was defined in
\citet{maximin} as the effect that is common to all sub-groups of
data and a simple estimator based on subsampling of the data was
proposed in \citet{magging}. However, the estimators for maximin
effects proposed so far just yield
point estimators but  we are interested here in
confidence intervals. While we are mostly dealing with low-dimensional
data where the sample size exceeds the number of samples, the results
could potentially be extended to high-dimensional regression using
similar ideas as proposed for example in 
\citet{zhang2011confidence} or \citet{van2013asymptotically}  for the estimation
of optimal linear regression effects for high-dimensional data.

\subsection{Model and notation}
We  first present a model for inhomogeneous data as considered in \citet{maximin}.
Specifically, we look at a special case where the data are split into
several known groups $g=1,\ldots,G$. In each group  $g$, we assume a linear model of the form
\begin{equation}\label{model}
 Y_{g} = \bold{X}_{g} b_{g}^0 + \varepsilon_{g},
\end{equation}
where $Y_{g}$ is a $n$-dimensional response vector of interest, $b_{g}^0$ a
deterministic $p$-dimensional regression parameter vectors and $\bold{X}_g$
a $n\times p$-dimensional design matrix containing in the columns the $n$
observations of $p$ predictor variables.  The noise contributions  $\varepsilon_{g}$  are
assumed to be independent with distribution $\mathcal{N}_n(0,\sigma^2
\mbox{Id}_n)$.  We assume the sample size $n$ to be identical in each
group. Generalizations to varying-coefficient models
\citep{hastie1993varying,fan1999statistical} are clearly possible but notationally more cumbersome. Inhomogeneity is caused by the different parameter vectors in the group. We define $\bold{X}$ as the row-wise concatenation of the design matrices 
$ \bold{X}_{1} ,  \bold{X}_{2} , \ldots, \bold{X}_{G} $ and assume
that the groups are known, that is we know which observations belong to the groups $g=1,\dots,G$, respectively.
 For the distribution of $\bold{X}_{g}$, $g=1,\ldots,G$ we consider different scenarios.

\paragraph{Scenario 1.} Random design. The observations of the
predictor variables are independent samples of an unknown multivariate distribution $F$ with finite fourth moments. We assume this distribution to be common across all groups $g=1,\ldots,G$. 

\paragraph{Scenario 2.} Random design in each group. The  observation in each group are independent samples of an unknown distribution $F_g$ with finite fourth moments. Observations in different groups are independent. The distribution $F_g$ may be different in different groups. 

\medskip

In the following if not mentioned otherwise we assume Scenario 1. The generalization to Scenario 2 is to a large extent only notational.

\subsection{Aggregation}

The question arises how the inhomogeneity of the optimal regression across groups is taken into account when trying to estimate the relationship between the predictor variables and the outcome of interest. Several known alternatives such as mixed effects models \citep{pinheiro2000mixed}, mixture models \citep{mclachlan2004finite} and clusterwise regression models \citep{desarbo1988maximum} are possibilities and are useful especially in cases where the group structure is unknown. They are at the same time computationally quite demanding. 

A computationally attractive alternative (especially for the discussed case of known groups but also more generally) is to estimate the optimal regression coefficient separately in each group, which are either known (as assumed in the following) or sampled in some appropriate form \citep{maximin}.  
As estimates for the $b_g^0$ we use in the following standard least squares estimators
\begin{equation*}
 \hat b_g = \argmin_{b \in \mathbb{R}^p} \| Y_g - \bold{X}_g b \|_2^2.
\end{equation*}
The restriction to this estimator is only for the purpose of simplicity. 
Regularization can be added if necessary but the essential issues are
already visible for least-squares estimation.  

Now a least-squares estimator is obtained in each group of data and
the question is how these different estimators can be aggregated.
The simplest and perhaps most widely-used aggregation scheme is bagging (bootstrap aggregation), as proposed by \citet{brei96}, where the aggregated estimator is given by 
\begin{equation} \label{eq:bagging}
{\bf Bagging}: \quad \hat b := \sum_g w_g \hat{b}_g,\qquad \mbox{where  } w_g=\frac 1 G \;\; \forall g=1,\ldots,G.
\end{equation}
If the data from different groups originate from an independent sampling mechanism, the bagging is a useful aggregation scheme. In particular, computing the bagged estimator is computationally more attractive than computing a single least-squares estimator as it allows the data to be split up into distinct subsets and processed independently before the aggregation step. For inhomogeneous data, the variability of the estimates $\hat b_g$ for $g=1,\ldots,G$ allows to gain some insight into the nature of the inhomogeneity. However, as argued in \cite{magging},  averaging is the wrong aggregation mechanism for inhomogeneous data.

\subsection{Maximin effect and magging} 
For inhomogeneous data, instead of looking for an estimator that works best on average, \citet{maximin} proposed to aim to maximize the minimum explained variance across several settings $g=1,\dots,G$. To be more precise, in our setting,
\begin{equation*}
	b_{\text{maximin}} := \argmax_{b \in \mathbb{R}^p} \min_{g
          =1,\dots,G} V(b,b_g^0),
\end{equation*}
where $V(b,b_g^0)$ is the explained variance in group $g$ (with true
regression vector $b_g^0$) when using a
regression vector $b$. That is 
\begin{align*} V(b,b_g^0) & :=\mathbb{E}\| Y_g \|_2^2-\mathbb{E} \|Y_g-X_g b
\|_2^2 \\ & =  2 b^t \Sigma^0  b_g^0 - b^t \Sigma^0 b   ,\end{align*}
where $\Sigma^0 := \mathbb{E} \hat \Sigma$ with $\hat \Sigma := (nG)^{-1}\bold{X}^t \bold{X}$ is the
sample covariance matrix.
In words, the maximin effect is defined as the
estimator that maximises the explained variance in the most adversarial
scenario (``group''). In this sense, the maximin effect is the
effect that is common among all groups in the data and ignores the
effects that are present in some groups but not in others.
It was shown in \citet{maximin} that the definition above is
equivalent to 
\begin{equation*}
	b_{\text{maximin}} = \argmin_{b \in \text{CVX}(B^0)} b^t \Sigma^0 b,
\end{equation*}
where $B^0= (b_1^0,\ldots,b_G^0) \in \mathbb{R}^{p \times G}$ the
matrix of the regression parameter vectors and $CVX(B^0)$ denotes the
closed convex hull of the $G$ vectors in $B^0$.
The latter definition motivates \textbf{ma}ximin \textbf{agg}regat\textbf{ing}, or \textbf{magging} \citep{magging}, which is the convex combination that minimizes the $\ell_2$-norm of the fitted values:
\begin{align*}
	\text{\textbf{Magging: }} \hat b := \sum_{g=1}^G \alpha_g
        \hat b_g,\qquad 
	\mbox{where } \;\; \alpha &:= \argmin_{\alpha \in C_G} \| \sum_{g=1}^G \alpha_g \bold{X} \hat b_g \|_2
	\;\; \mbox{and }  \\ C_g &:= \{ \alpha \in \mathbb{R}^G : \min_{g} \alpha_g \ge 0 \text{ and } \sum_g \alpha_g =1  \}	
\end{align*}
The magging regression vector is unique if $\bold{X}^t \bold{X}$ is
positive definite. Otherwise, we can only identify the prediction
effect $\bold{X} b_{\text{maximin}}$ and the solution above is meant
to be any member of the feasible set of solutions. 
To compute the estimator, the dataset is split into several smaller
datasets and we assume here that the split separates the data into
already known groups. After computing estimators on all of these
groups separately, possibly
in parallel, magging can be used to find common effects of all
datasets. This is in particular interesting if there is inhomogeneity in
the data. For known groups, as in our setting, magging can be interpreted as the plug-in
estimate of the maximin effect. 

In the following we need additional notation. For $B := (b_1,\ldots.,b_G)
\in \mathbb{R}^{p \times G}$ and for $\Sigma \in \mathbb{R}^{p \times p}$ positive definite define
\begin{equation*}
M_\Sigma(B) := \argmin_{b \in \text{CVX}(B)} b^t \Sigma b
\end{equation*}
We obtain the original definition of the magging
estimator for $M_{\hat \Sigma}(\hat B)$ with $\hat B = (\hat b_1,\ldots, \hat b_G
)$ and the maximin effect with
$M_{ \Sigma^0}( B^0)$.
 
\subsection{Novel contribution and organization of the paper}
So far only point estimators of maximin effects have been proposed in
the literature. In Section~\ref{theory} we discuss an asymptotic
approach to construct confidence regions for the  maximin
effect. Specifically, we calculate the asymptotic distribution of
$\sqrt{n} (M_{\hat \Sigma}(\hat B)-M_{ \Sigma^0}( B^0))$ and derive
corresponding asymptotically valid confidence regions. This
gives us (asymptotically) tight confidence regions and will shed more
light on the (asymptotic) nature of the  fluctuations of the magging
estimator. We evaluate the actual coverage of this approximation on
simulated datasets in Section~\ref{practice}. The proofs of the
corresponding theorems and an alternative non-asymptotic approach can
be found in the appendix. The advantages and disadvantages of the
approaches are discussed in Section~\ref{sec:discussion}.

\section{Confidence intervals for maximin effects}\label{theory}

In Scenario 1, the random design of the predictor variables is
identical across all groups of data. For fixed $G$ and $n \rightarrow \infty$, we
can then use the delta method to
derive the asymptotic distribution of the scaled difference between
the true and estimated magging effects \[ \sqrt{n}(M_{\hat \Sigma}(\hat
B)-M_{\Sigma^0}(B^0)).\] 
This in turn allows to construct confidence intervals for the true
maximin effects. Let $W(\hat B,\hat \Sigma)$ be
a consistent estimator of the (positive definite) variance of the Gaussian distribution
$\lim_{n \rightarrow \infty} \sqrt{n}(M_{\hat \Sigma}(\hat B)-M_{\Sigma^0}(B^0))$.
Let $\alpha > 0$. Choose $\tau$ as the $(1-\alpha)$-quantile of the $\chi_p^2$-distribution.
Define then a confidence region as 
\begin{equation}\label{eq:CI}
\bold{C}(\hat \Sigma,\hat B) := \{ M \in \mathbb{R}^p :   (M_{\hat \Sigma}(\hat B)-M)^t W(\hat B,\hat \Sigma)^{-1}  (M_{\hat \Sigma}(\hat B)-M) \le \frac{\tau}{n} \}
\end{equation}
The definition of  $W(\hat B,\hat \Sigma)$ is deferred to the
appendix, Section \ref{defs}. We will show in the following that we
obtain asymptotically valid confidence intervals with this
approach. For simplicity, we work with Scenario~1 here and assume that the
noise contributions  $\varepsilon_{g}$ in equation \eqref{model} are
 independent with distribution $\mathcal{N}_n(0,\sigma^2 \mbox{Id}_n)$. Furthermore,  each $\bold{X}_{g} \in
\mathbb{R}^{n \times p}$ is assumed to have full rank, requiring $p \le n$. Though the
framework for the result is a Gaussian linear model, it can be easily extended to more general settings.

The following theorem describes the coverage properties of the
confidence interval~\eqref{eq:CI}. In the following, for $x,y \in
\mathbb{R}^p$ and $\Sigma \in \mathbb{R}^{p \times p}$ positive definite define $\langle x,y
\rangle_\Sigma := x^t \Sigma y$.
\begin{theorem}\label{coverage}
 Let $\Sigma^0$ be positive definite. Let $M_{\Sigma^0}(B^0) = \sum_{g =1}^G \alpha_g b_g^0 $ with $\alpha_g \ge 0$,
 $\sum_{g=1}^G \alpha_g =1$ and let this representation be unique. Let
 $| \{ g: \alpha_g \neq 0 \} | > 1$. Suppose that the hyperplane
 orthonormal to the maximin effect contains only ``active'' $b_g^0$, i.e. $
 \{ b_g^0 : g =1,...,G \} \cap \{ M \in
 \mathbb{R}^p : \langle M-M_{\Sigma^0}(B^0),M_{\Sigma^0}(B^0)  
 \rangle_{\Sigma^0}=0  \}  \subset \{ b_g^0 : \alpha_g \neq 0 \}$.
Then
\begin{equation*}
\lim_{n \rightarrow \infty}\mathbb{P} [ M_{ \Sigma^0}( B^0) \in \bold{C}(\hat \Sigma,\hat B)] = 1-\alpha.
\end{equation*}
\end{theorem}
In other words, the set defined in~\eqref{eq:CI} is an asymptotically
valid confidence region for $M_{ \Sigma^0}(B^0)$ under the made
assumptions.  If the true coefficients $b_g^0$  in each group are drawn from a multivariate density, then the assumptions are fulfilled with probability one.

 The special case $| \{ g: \alpha_g \neq 0 \} | = 1$ is excluded, as the
magging estimator is identical to a  solution in one individual group
in this case, which is equivalent to 
$M_{\hat \Sigma}(\hat B)= \hat b_g $ for a $g \in \{1,\ldots,G
\}$, up to an asymptotically negligible set. This case is mainly excluded
for notational reasons. The assumptions of Theorem \ref{coverage} guarantee that the derivative of
magging $M_\Sigma(B)$ exists and is continuous at $B^0$ and
$\Sigma^0$. If the latter condition is violated, it is still possible to obtain
asymptotic bounds in the more general setting, as $\lim_{n \rightarrow
  \infty} \sqrt{n}(M_{\hat \Sigma}(\hat B)-M_{\Sigma^0}(B^0))$ is still
subgaussian. We explore the violation of these assumptions with simulation
studies in the next section. The proof of Theorem \ref{coverage} is an
application of Slutsky's Theorem, combined with the following result about
the asymptotic variance of the magging estimator.
\begin{theorem}\label{maintheorem}
Let the assumptions of Theorem~\ref{coverage} be true. 
Then, for $n \rightarrow \infty$,
 \begin{equation}\label{toprove}
 \sqrt{n} \left(M_{\hat \Sigma}(\hat B)-M_{\Sigma^0}( B^0)  \right)  \rightharpoonup  \mathcal{N} \Big( 0, \sigma^2 \sum_{g \in A(B^0,\Sigma^0)} \Dif_g^t M_{\Sigma^0}(B^0) \Sigma^{-1} \Dif_g M_{\Sigma^0}(B^0)+V(B_{A(B^0,\Sigma^0)}^0,\Sigma^0) \Big).
 \end{equation}
\end{theorem}

Here, $\Dif_g$ denotes the differential in direction $b_g$. This derivative is
calculated in the appendix, see Section~\ref{defs}. The set $A(B,\Sigma)
\subset \{ 1,\ldots,G\}$ denotes indices $g$ for which $b_g$ has nonvanishing
coefficient $\alpha_g$ in one of the convex combinations
$M_\Sigma(B)=\sum_{g=1,\ldots,G} \alpha_g b_g$ with $\alpha_g \ge 0$,
$\sum_{g=1,\ldots,G} \alpha_g =1$. Note that by the assumptions of Theorem
\ref{coverage} this convex combination is unique for
$M_{\Sigma^0}(B^0)$. The definition of $V(B_{A(B,\Sigma)},\Sigma)$ is
somewhat lengthy and can be found in the appendix, Section~\ref{defs}.

The first summand in the variance in formula \eqref{toprove} is due to fluctuations of the estimator of $B^0$, the second summand is due to fluctuations of the estimator of $\Sigma^0$. 
If $\Sigma^0$ is known in advance, we can use $\hat \Sigma := \Sigma^0$ and in the theorem above $V = 0$. Table~\ref{illustrate} is an illustration of Theorem \ref{maintheorem}.

\begin{figure}
\includegraphics[scale=0.3]{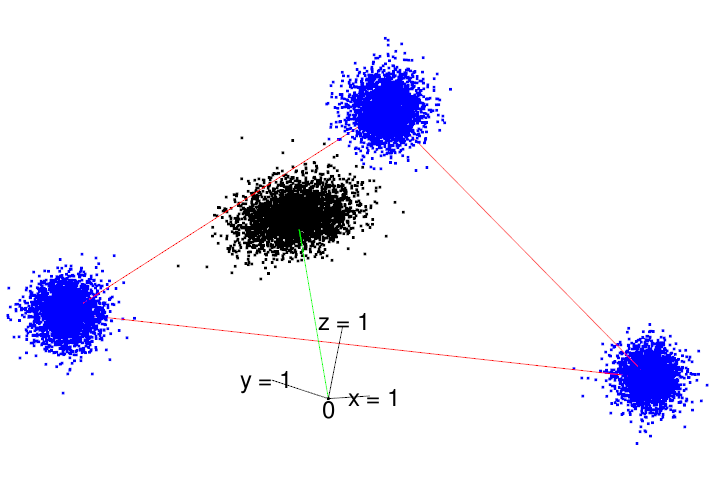}
\includegraphics[scale=0.3]{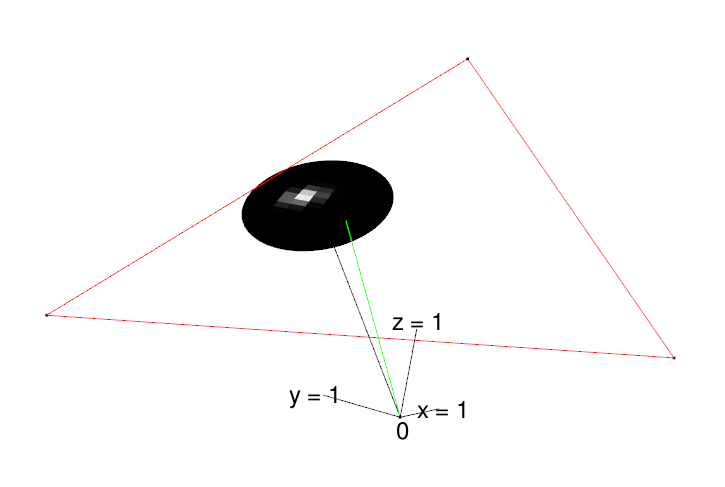}
\caption{\it An illustration of 
  Theorems~\ref{coverage} and~\ref{maintheorem}. On the left
  hand side the blue dots represent 3000 realizations of $\hat b_g$, $g=1,2,3$
  with dimension $p=3$. The black dots are the
  corresponding magging estimates $M_{\hat \Sigma}(\hat B)$. The green line
  indicates the true maximin effect $M_{\Sigma^0}(B^0)$. On the right hand side,
  the black line indicates one of the  $M_{\hat \Sigma}(\hat B)$ with the corresponding approximate 95\%-confidence region calculated with the terms of equation \eqref{eq:CI}.}\label{illustrate}
\end{figure}

\section{Numerical Examples}\label{practice}

The aim of this section is to evaluate the actual coverage of the
approximate confidence regions as defined above. We study several
examples. They have in common that the entries in $\bold X$ are
i.i.d.\ $\mathcal{N}(0,1)$. Furthermore the $\varepsilon_g$
 are
i.i.d.\ $\mathcal{N}(0,\text{Id}_n)$ and independent of $\bold X$. The
tables show the coverage of the
  true maximin effect $M_{\Sigma^0}(B^0)$ by the proposed 95\% confidence
regions. We calculate the confidence intervals only for $p<n$ scenarios as
long as least squares estimators are used (Tables 1-3), while the
  case of $p\ge n$ is covered in tables~4 and~5 by the use of a ridge
  penalty. All simulations were run $1000$ times.

In the setting of Table~\ref{tablecoverage} all assumptions of Theorem~\ref{coverage} are
satisfied. As expected, for large $p$ the convergence of the actual
coverage seems to be slower. Note that for validity of Theorem~\ref{coverage} it is not
necessary that $p=G$, as we have asymptotically tight coverage for all $1 <
G \le p$.

% latex table generated in R 3.1.2 by xtable 1.7-4 package
% Tue Dec 23 19:04:56 2014
\begin{table}[ht]
\centering
\begin{tabular}{rrrrrrrrrr}
  \hline
 & $n= 5$ & 10 & 15  & 100 & 200 & 500 & 1000 & 2000 & 4000\\ 
  \hline
$p=3$ & 0.70 & 0.78 & 0.82 & 0.92 & 0.94 & 0.95 & 0.94 & 0.94 & 0.95 \\ 
  5 &  & 0.69 & 0.76 & 0.90 & 0.93 & 0.95 & 0.94 & 0.95 & 0.95 \\ 
  10 &  &  & 0.62 & 0.84 & 0.88 & 0.94 & 0.95 & 0.96 & 0.94 \\ 
  15 &  &  &  & 0.78 & 0.85 & 0.93 & 0.92 & 0.95 & 0.95 \\ 
  20 &  &  &  & 0.72 & 0.83 & 0.90 & 0.91 & 0.95 & 0.94 \\ 
  40 &  &  &  & 0.54 & 0.63 & 0.79 & 0.88 & 0.91 & 0.94 \\ 
  80 &  &  &  & 0.57 & 0.38 & 0.50 & 0.74 & 0.85 & 0.92 \\
   \hline
\end{tabular}
\caption{$b_g^0 = e_g$, $g=1,\ldots,G=p$, where the $e_g$ denote the vectors of the standard basis, $1000$ iterations. The  coverage can be seen to be approximately correct if $n$ is sufficiently large.}\label{tablecoverage}
\end{table}

% latex table generated in R 3.1.2 by xtable 1.7-4 package
% Tue Dec 23 13:52:09 2014
\begin{table}[H]
\centering
\begin{tabular}{rrrrrrrrrr}
  \hline
 & $n= 5$ & 10 & 15  & 100 & 200 & 500 & 1000 & 2000 & 4000\\ 
  \hline
$p=3$ & 0.64 & 0.84 & 0.91 & 0.97 & 0.96 & 0.82 & 0.98 & 0.96 & 0.97 \\ 
  5 &  & 0.61 & 0.79 & 0.99 & 0.97 & 0.88 & 0.82 & 0.91 & 1.00 \\ 
  10 &  &  & 0.23 & 0.99 & 0.99 & 1.00 & 0.99 & 0.93 & 0.98 \\ 
  15 &  &  &  & 0.99 & 0.99 & 1.00 & 1.00 & 0.99 & 0.99 \\ 
  20 &  &  &  & 0.99 & 1.00 & 0.99 & 1.00 & 1.00 & 0.99 \\ 
  40 &  &  &  & 0.94 & 1.00 & 1.00 & 1.00 & 1.00 & 1.00 \\ 
  80 &  &  &  & 0.00 & 1.00 & 1.00 & 1.00 & 1.00 & 1.00 \\ 
   \hline
\end{tabular}
\caption{$b_g^0 = e_1+z_g e_2$, $g=1,\ldots,G=p$, $z_g \sim \mathcal{N}(0,1)$
  independent. The assumptions are violated, yielding
  too conservative confidence intervals. The $0.00$ at $n=100$, $p=80$ is
  due to a large bias of $M_{\hat \Sigma}(\hat B)$ towards $0$. For larger
  $n$, however, this bias quickly vanishes and we get the desired coverage
  (starting at approximately $n =120$).}\label{biaseffect}
\end{table}

% latex table generated in R 3.1.2 by xtable 1.7-4 package
% Tue Dec 23 14:09:10 2014
\begin{table}[H]
\centering
\begin{tabular}{rrrrrrrrrr}
  \hline
 & $n= 5$ & 10 & 15  & 100 & 200 & 500 & 1000 & 2000 & 4000\\ 
  \hline
$p=3$ & 0.76 & 0.87 & 0.90 & 0.99 & 0.99 & 0.99 & 1.00 & 1.00 & 1.00 \\ 
  5 &  & 0.65 & 0.78 & 1.00 & 1.00 & 1.00 & 1.00 & 1.00 & 1.00 \\ 
  10 &  &  & 0.33 & 1.00 & 1.00 & 1.00 & 1.00 & 1.00 & 1.00 \\ 
  15 &  &  &  & 0.99 & 1.00 & 1.00 & 1.00 & 1.00 & 1.00 \\ 
  20 &  &  &  & 0.99 & 1.00 & 1.00 & 1.00 & 1.00 & 1.00 \\ 
  40 &  &  &  & 0.93 & 1.00 & 1.00 & 1.00 & 1.00 & 1.00 \\ 
  80 &  &  &  & 0.00 & 1.00 & 1.00 & 1.00 & 1.00 & 1.00 \\
   \hline
\end{tabular}
\caption{$b_g^0 =  e_1$,
  $g=1,\ldots,G=[0.8p]$. The assumptions are again violated
  and coverage is too high.  At $p=80$ and $n=100$ we observe the same
  effect as in Table \ref{biaseffect}.
  In this scenario the estimated
  confidence regions can become arbitrarily large. This stems from the fact
that if some of the  $\hat b_g$ corresponding to $A(\hat B, \hat \Sigma)$
are very  close, the estimated variance of magging may become
large. In this setting a different approach, for example as discussed in Section~\ref{relaxation} makes more sense.}\label{violated}
\end{table}

% latex table generated in R 3.1.2 by xtable 1.7-4 package
% Tue Dec 23 19:25:40 2014
\begin{table}[H]
\centering
\begin{tabular}{rrrrrrrrrr}
  \hline
 & $n= 5$ & 10 & 15  & 100 & 200 & 500 & 1000 & 2000 & 4000\\ 
  \hline
$p=3$ & 0.71 & 0.77 & 0.84 & 0.92 & 0.94 & 0.96 & 0.95 & 0.94 & 0.93 \\ 
  5 & 0.74 & 0.69 & 0.76 & 0.90 & 0.94 & 0.94 & 0.95 & 0.95 & 0.95 \\ 
  10 & 0.55 & 0.70 & 0.60 & 0.86 & 0.88 & 0.93 & 0.94 & 0.94 & 0.95 \\ 
  15 & 0.52 & 0.53 & 0.70 & 0.77 & 0.86 & 0.91 & 0.94 & 0.95 & 0.95 \\ 
  20 & 0.53 & 0.48 & 0.52 & 0.73 & 0.81 & 0.89 & 0.93 & 0.92 & 0.94 \\ 
  40 & 0.40 & 0.47 & 0.37 & 0.52 & 0.62 & 0.81 & 0.87 & 0.90 & 0.94 \\ 
  80 & 0.20 & 0.40 & 0.37 & 0.56 & 0.38 & 0.52 & 0.72 & 0.84 & 0.90 \\ 
   \hline
\end{tabular}
\caption{$b_g^0 =  e_g$, $g=1,\ldots,G=p$. The diagonal elements of $\hat
  \Sigma$ and $\hat \Sigma_g$ where increased by a value $10^{-4}$ in order to make
 them invertible and not too ill-conditioned for $n \le p $. Again, coverage is approximately correct for $n$ sufficiently large.}\label{ridge}
\end{table}

% latex table generated in R 3.1.2 by xtable 1.7-4 package
% Tue Dec 23 19:29:57 2014
\begin{table}[H]
\centering
\begin{tabular}{rrrrrrrrrr}
  \hline
 & $n= 5$ & 10 & 15  & 100 & 200 & 500 & 1000 & 2000 & 4000\\ 
  \hline
$p=3$ & 41.70 & 2.97 & 1.59 & 0.59 & 0.53 & 0.49 & 0.47 & 0.47 & 0.46 \\ 
  5 & 831.50 & 13.52 & 4.83 & 0.42 & 0.34 & 0.30 & 0.28 & 0.26 & 0.26 \\ 
  10 & 6.56 & 1935.77 & 27.78 & 0.29 & 0.20 & 0.16 & 0.14 & 0.13 & 0.12 \\ 
  15 & 0.29 & 19.83 & 3844.87 & 0.26 & 0.16 & 0.12 & 0.10 & 0.09 & 0.08 \\ 
  20 & 0.08 & 4.25 & 41.04 & 0.29 & 0.15 & 0.09 & 0.08 & 0.07 & 0.06 \\ 
  40 & 0.01 & 0.04 & 4.61 & 2.71 & 0.16 & 0.07 & 0.05 & 0.04 & 0.03 \\ 
  80 & 0.00 & 0.00 & 0.01 & 205.85 & 1.09 & 0.06 & 0.03 & 0.02 & 0.02 \\
   \hline
\end{tabular}
\caption{This table shows the average maximum eigenvalues of the estimated covariance matrix of $\sqrt{n}(M_{\Sigma^0}(B^0)-M_{\hat \Sigma}(\hat B))$, analogous to Table \ref{ridge}.}\label{ridge2}
\end{table}

In Table~\ref{biaseffect} and Table~\ref{violated} we explore the
violation of one of the assumptions in Theorem~\ref{coverage}. The maximin
effect is $M_{\Sigma^0}(B^0)=(1,0,0\ldots)$, and the convex combination
$M_{\Sigma^0}(B^0) = \sum_{g=1}^G \alpha_g b_g^0$ with $\alpha_g \ge 0$, $\sum
\alpha_g =1$ is not unique. In both cases, this seems to lead to too
conservative confidence regions. Generally, in these settings the
difficulty arises from the fact that the derivative of $M_{\Sigma}(B)$
does not exist at $M_{\Sigma^0}(B^0)$. As a result, the fluctuations of
$\lim_n \sqrt{n}(M_{\Sigma^0}(B^0)-M_{\hat \Sigma}(\hat B))$  - provided
that this limit exists - are not necessarily
Gaussian anymore.

In the last simulation, depicted in Table~\ref{ridge} the $\hat b_g$,
$g=1,\ldots,G$ were not calculated by ordinary
least squares but ridge regression. The diagonal elements of $\hat
  \Sigma$ and $\hat \Sigma_g$ where increased by a value $10^{-4}$ in order
  to make them invertible and not too ill-conditioned for $n \le p $. Apart from that we used the same
setting as in Table~\ref{tablecoverage}. As in Table~\ref{tablecoverage}, for large $n$ the coverage seems to be
(approximately) correct but severe undercoverage can still occur for
$n\ll p$. In these high-dimensional settings, the tuning ridge parameter
would need to be better adjusted for a useful balance between bias and
variance and the bias of the ridge penalty would have to be
adjusted for, something which is beyond the current scope. In Table~\ref{ridge2} the
corresponding maximum eigenvalues of the estimated variance of
$\sqrt{n}(M_{\Sigma^0}(B^0)-M_{\hat \Sigma}(\hat B))$ were plotted, each entry being the
average over all
1000 runs. We observe a spike for $p=n$. This peaking is similar to a
related effect in ridge and lasso regression. Specifically, for fixed $p$ and varying $n$, the
norm of the regression estimate is growing as $n$ is increased,  reaching its peak at approximately
$p=n$ while then decreasing again as the solution converges towards
the true parameter as $n$ grows very large.

\section{Discussion}\label{sec:discussion}

 We derived the asymptotic distribution of the magging estimator and proposed
asymptotically tight and valid confidence regions for the maximin
effect. The corresponding theorems requires a rather weak assumption on
the true regression coefficients $b_1^0,\ldots,b_G^0$. However, if this
assumption is not satisfied, as studied in simulations, the resulting
confidence regions seem to become too conservative. Especially when all of the
``active'' vectors $\{ \hat b_g$ : $g \in A(\hat \Sigma, \hat B) \}$ are very
close to each other, the proposed confidence regions tend to become large. Furthermore,
in this scenario the magging estimator may suffer from a large bias. Then it may make more sense to use an approach based on
relaxation. Such an approach is outlined in the appendix in
Section~\ref{relaxation} and it would also allow for non-asymptotic
confidence intervals at the price of coverage probabilities well above
the specified level. The proposed asymptotic confidence interval on
the other hand is arguably more intuitive and yields in most scenarios
tight bounds for
large sample sizes.

\bibliographystyle{apalike} 
\bibliography{references}

\section{Appendix}

The structure is as follows: The first part is devoted to the most
important definitions and explicit formulas which were omitted in the main
section of the paper. The second part contains the proof of
Theorem~\ref{maintheorem} and several lemmata. The third part
contains the proof of Theorem~\ref{coverage}. Finally, the last part
contains a relaxation-based idea to construct confidence intervals for
maximin effects.

\subsection{Definitions and formulas}\label{defs}

\begin{definition}{$A(B,\Sigma)$}

The set $A(B,\Sigma) \subset \{ 1,\ldots,G\}$ denotes indices $g$ for which $b_g$ has nonvanishing
coefficient $\alpha_g$ in one of the convex combinations
$M_\Sigma(B)=\sum_{g=1,\ldots,G} \alpha_g b_g$ with $\alpha_g \ge 0$,
$\sum_{g=1,\ldots,G} \alpha_g =1$. Note that by the assumptions of Theorem \ref{coverage}
or Theorem \ref{maintheorem} the $\alpha_g$ are unique for $M_{\Sigma^0}(B^0)$.

\end{definition}

\begin{definition}{$W( B, \Sigma)$} 

$W(\hat B,\hat \Sigma)$ is a consistent estimator of the variance of $\lim_n \sqrt{n} (M_{\hat \Sigma}(\hat B)- M_{ \Sigma^0}( B^0) )$, see proof of Theorem \ref{maintheorem}.

\begin{equation*}
W( B, \Sigma) = \sigma^2 \sum_{g \in A( B, \Sigma)} \Dif_g^t M_{ \Sigma}( B_{A( B, \Sigma)}) { \Sigma}^{-1} \Dif_g M_{\Sigma }( B_{A( B,\Sigma)})+V( B_{A( B, \Sigma)}, \Sigma) 
\end{equation*}
Definitions and explicit formulas of these terms can be found below. We estimate $\Sigma^0$ by $\hat \Sigma = \frac{1}{nG} \bold{X}^t \bold X
 $. $\Dif_g^t M_\Sigma(B)$ denotes the derivative of $M_\Sigma(B)$ with
 respect to $b_g$.

\end{definition}
\paragraph{Explicit formula for $V ( \hat B_{A(\hat B,\hat \Sigma)},\hat
  \Sigma)$.} (Compare with Lemma~\ref{matcalc})

Consistent estimator of the additional variance of $\lim_n \sqrt{n}
(M_{\hat \Sigma}(\hat B)- M_{ \Sigma^0}( B^0) )$ ``caused''  by not knowing $\Sigma^0$, see proof of Theorem \ref{maintheorem} and Lemma \ref{matcalc}.

\begin{equation*}
 V ( \hat B_{A(\hat B,\hat \Sigma)},\hat \Sigma) =  \hat D ( \hat D ^t \hat \Sigma \hat D )^{-1} \hat D^t \hat C \hat D (\hat D^t \hat \Sigma \hat D )^{-1} \hat D^t,
\end{equation*}
where  $\hat{C}$ is the empirical covariance matrix of the $p$-dimensional
vectors $\frac{1}{\sqrt{G}}\bold{X}_{k \cdot}^t \bold{X}_{k \cdot}M_{\hat \Sigma}(\hat B)$, $k=1,\ldots,(nG)$.
Furthermore, with $\tilde B = \hat B_{A(\hat B,\hat \Sigma)}$, $G' = | A(\hat B,\hat \Sigma) |$:

\begin{equation*}
\hat D := (\tilde b_2,\ldots,\tilde b_{G'}) - (\tilde b_1,\ldots,\tilde b_1).
\end{equation*}

\paragraph{Explicit formula for  $\Dif_g M_{\hat \Sigma}(\hat B_{A(\hat
    \Sigma, \hat B)})$.}(Compare with Lemma \ref{diff1})

Let us again write $\tilde B = \hat B_{A(\hat B,\hat \Sigma)}$, $G' = | A(\hat B,\hat \Sigma) |$,
\begin{align*}
 \Dif_g M_{\hat \Sigma}(\hat B_{A(\hat \Sigma, \hat B)}) = &-  \frac{\| M_{\hat \Sigma}(\tilde B) \|_{\hat \Sigma}}{\| (\text{Id}-\hat{\text{PA}}^{(g)})\tilde b_g \|_{\hat \Sigma}} 
 \frac{  (\text{Id}-\hat{\text{PA}}^{(g)})\tilde b_g}{\| (\text{Id}-\hat{\text{PA}}^{(g)})\tilde b_g \|_{\hat \Sigma}} \frac{M_{\hat \Sigma}(\tilde B)^t}{\| M_{\hat \Sigma}(\tilde B) \|_{\hat \Sigma}}\hat \Sigma  \nonumber \\
 &+\frac{\| (\text{Id}- \hat{\text{PA}}^{(g)})M_{\hat \Sigma}(\tilde B) \|_{\hat \Sigma}}{\| (\text{Id}-\hat{\text{PA}}^{(g)})\tilde b_g \|_{\hat \Sigma}} \hat \Pi_{\tilde B}. 
\end{align*}
Here, $\hat{\text{PA}}^{(g)}$ denotes the affine projection on the smallest affine space containing $\tilde b_1,\ldots,\tilde b_{g-1},\tilde b_{g+1},\ldots,\tilde b_{G'}$. Let $\Pi_{\tilde B} \in \mathbb{R}^{p \times p}$ denote the projection on $\langle  \tilde b_2- \tilde b_1,\ldots, \tilde b_{G'}-\tilde b_1 \rangle^\perp$. These geometric definitions are meant with respect to the scalar product $\langle x,y \rangle_{\hat \Sigma} = x^t \hat \Sigma y$.

\subsection{Proof of Theorem \ref{maintheorem}}

\begin{proof}
 The proof is based on the delta method. As $\hat B \rightharpoonup B^0$ and
 $\hat \Sigma \rightharpoonup \Sigma^0$, by Lemma~\ref{continuity}, $A(B^0,\Sigma^0)$ = $A(\hat B,\hat \Sigma)$ up to an
 asymptotically negligible set. Hence $ M_{\Sigma^0}(B^0) =
 M_{\Sigma^0}(B_{A(B^0,\Sigma^0)}^0)$ and $ M_{\hat \Sigma}(\hat B) = M_{\hat
   \Sigma}(\hat B_{A(B^0,\Sigma^0)}) $ up to an asymptotically negligible
 set. So without loss of generality let us assume (without changing the
 definition of $\hat \Sigma$) that $ A(B^0,\Sigma^0)= A(\hat B,\hat \Sigma)=
 \{1,\ldots,G\} $, and hence $ B^0=B_{ A(B^0,\Sigma^0)}^0 $, $\hat  B=\hat B_{ A(\hat B,\hat \Sigma)} $. By Lemma \ref{diff1} and Lemma \ref{diff2}, $M_{\Sigma}(B)$ is continuously differentiable in a neighborhood of $B^0$ and $\Sigma^0$. Using Taylor in a neighborhood of $B^0$ and $\Sigma^0$ we can write
 
 \begin{align*}
  \sqrt{n} \left(M_{\hat \Sigma}(\hat B)-M_{\Sigma^0}(B^0)  \right) =& \Dif_{B }M_{ \Xi }(\xi) \sqrt{n}(   \hat B-B^0 )\\
  &+ \Dif_{\Sigma }M_{ \Xi }(\xi) \sqrt{n}(   \hat \Sigma-\Sigma^0 )+\scriptstyle\mathcal{O} \textstyle_\mathbb{P} (1)  \\
  =& (\Dif_B M_\Xi(\xi)-\Dif_B M_{\Sigma^0}( B^0)) \sqrt{n}(  \hat B - B^0 ) \\
  &+(\Dif_\Sigma M_\Xi(\xi)-\Dif_\Sigma M_{\Sigma^0}( B^0)) \sqrt{n}(  \hat \Sigma - \Sigma^0 ) \\
  &+\Dif_B M_{\Sigma^0}( B^0)\sqrt{n}(  \hat B- B^0 ) \\
  &+\Dif_\Sigma M_{\Sigma^0}( B^0)\sqrt{n}(  \hat \Sigma- \Sigma^0 ) +\scriptstyle\mathcal{O} \textstyle_\mathbb{P} (1),
  \end{align*}
with  $\xi = \gamma B^0+(1-\gamma) \hat B$ and $\Xi = \gamma
\Sigma^0+(1-\gamma) \hat \Sigma $ for some random variable $\gamma \in [0,1]$. 
We now want to show that the first and second term are negligible, and calculate the asymptotic Gaussian distributions of the last two terms. Furthermore we want to show that the last two terms are asymptotically independent. This guarantees that the variance of $ \lim_n \sqrt{n} \left(M_{\hat \Sigma}(\hat B)-M_{\Sigma^0}( B^0)  \right)$ is the sum of the variances of the two asymptotic Gaussian distributions.

Hence, to prove \eqref{toprove} it suffices to show: 
 \begin{enumerate}[(1)]
 \item $ \Dif_B M_\Xi(\xi)-\Dif_B M_{\Sigma^0}( B^0)= \scriptstyle\mathcal{O} \textstyle_\mathbb{P} (1)$
  \item $ \Dif_\Sigma M_\Xi(\xi)-\Dif_\Sigma M_{\Sigma^0}( B^0)= \scriptstyle\mathcal{O} \textstyle_\mathbb{P} (1)$
  \item $  \sqrt{n}(  \hat b_g - b_g^0 ) \rightharpoonup \mathcal{N}(0, \sigma^2 (\Sigma^0)^{-1}) $ for $g=1,\ldots,G$.
  \item $ \Dif_B M_{\Sigma^0}( B^0) \sqrt{n}(  \hat B - B^0 )  \rightharpoonup \mathcal{N} \left( 0, \sigma^2 \sum_{g \in A(B^0,\Sigma^0)} \Dif_g^t M_{\Sigma^0}(B^0) (\Sigma^0)^{-1} \Dif_g M_{\Sigma^0}(B^0) \right) $
  \item $\Dif_\Sigma M_{\Sigma^0}( B^0) \sqrt{n}(  \hat \Sigma - \Sigma^0 )\rightharpoonup \mathcal{N}(0,V(B^0,\Sigma^0)) $
  \item For $\delta_n := \sqrt{n} ( \hat B -B^0)$ and $ \Delta_n := \sqrt{n}(\hat \Sigma-\Sigma^0)$ we have $(\delta_n,\Delta_n)\rightharpoonup (\delta,\Delta)$ with $\delta_g$, $g=1,\ldots,G$ and $\Delta$ independent. 
 \end{enumerate}
 
\textbf{Part (1) and (2):} By Lemma \ref{diff1} and Lemma \ref{diff2} the 
 derivatives are continuous at $B^0$ and $\Sigma^0$ and $\hat \Sigma \rightarrow \Sigma^0$, $\hat B \rightarrow B^0$ in probability (which implies $\xi \rightarrow B^0$ and $\Xi \rightarrow \Sigma^0$ in probability).
 
\textbf{Part (3):} This is immediate, as under the chosen model, conditioned on $\bold{X}$,
\begin{equation*}
 \hat b_g \sim \mathcal{N}(b_g, \sigma^2(\bold{X}_g^t \bold{X}_g)^{-1})
\end{equation*}

and $\frac{1}{n} \bold{X}_g^t \bold{X}_g \rightarrow \Sigma$ in probability.

\textbf{Part (4):} Part (3) and a linear transformation.

\textbf{Part (5):} We defer this part to Lemma \ref{matcalc}.

\textbf{Part (6):} We saw the convergence of $\delta_n$ in part (3). The convergence of $\Delta_n$ is deferred to Lemma \ref{clt}. In the following we use the notation $\delta = (\delta_1,\ldots,\delta_G)$ and $\delta_n = (\delta_{n,1},\ldots,\delta_{n,G})$. For the asymptotic independence of part (6). we have to show that for any bounded continuous function $g$,
\begin{equation*}
\mathbb{E} g(\delta_n,\Delta_n) \rightarrow \int \int g(\delta,\Delta) \frac{(\text{det} \Sigma^0)^{G/2}}{(2 \pi \sigma^2)^{G/2}} \prod_{g=1}^{G} \exp \left(  -\delta_g^t \frac{\Sigma^0}{2\sigma^2} \delta_g \right) d \delta_1 \cdots d \delta_G \mathbb{P}[d \Delta].
\end{equation*}
In the following equation the inner integral is bounded by $2$, and for $n \rightarrow \infty$,
$\frac{1}{n} \bold{X}_g^t \bold{X}_g \rightarrow \Sigma^0$ in
probability. Hence, by dominated convergence on the inner and outer integral,
\begin{align*}
&\int \int | \prod_{g=1}^G \frac{\sqrt{ \text{det} \frac{1}{n} \bold{X}_g^t \bold{X}_g}}{(2 \pi \sigma^2)^{1/2}}  \exp \left(- \delta_{n,g}^t \frac{\bold{X}_g^t \bold{X}_g}{2n\sigma^2}  \delta_{n,g} \right)\\
 &- \prod_{g=1}^G \frac{\sqrt{ \text{det} \Sigma^0} }{(2 \pi \sigma^2)^{1/2}}  \exp \left(- \delta_{n,g}^t \frac{\Sigma^0}{2\sigma^2} \delta_{n,g} \right)| d \delta_{n,1} \cdots d \delta_{n,G} \mathbb{P}[d \Delta_n] \rightarrow 0.
\end{align*}
Using this,
\begin{equation*}\label{asymind}
\limsup_{n \rightarrow \infty} | \mathbb{E} g(\delta_n,\Delta_n)-g(\delta,\Delta_n) | = 0,
\end{equation*}
where $\delta$ is independent of $\Delta_n$, $\delta_g \sim
\mathcal{N}(0,\sigma^2 (\Sigma^0)^{-1})$ i.i.d.. Finally, with $\Delta$ independent of $\delta$, $\Delta \sim \lim_n \sqrt{n} (\hat \Sigma - \Sigma^0)$,
\begin{align*}
&\limsup_{n \rightarrow \infty} | \mathbb{E} g(\delta_n,\Delta_n) -\mathbb{E} g(\delta,\Delta)| \\
&=  \limsup_{n \rightarrow \infty} | \mathbb{E} g(\delta,\Delta_n)-\mathbb{E} g(\delta,\Delta) | \\
&=\limsup_{n \rightarrow \infty} |  \int \mathbb{E} [ (g(\delta,\Delta_n)-g(\delta,\Delta))| \delta ]\frac{(\text{det} \Sigma^0 )^{G/2}}{(2 \pi \sigma^2)^{G/2}}   \prod_{g=1}^G \exp \left(- \delta_{g}^t \frac{\Sigma^0}{2\sigma^2} \delta_{g} \right) d \delta_{1} \cdots d \delta_{G} | \\
&=0.
\end{align*}
In the second line we used equation \eqref{asymind}, in the last line we used dominated convergence and $\Delta_n \rightharpoonup \Delta$. This concludes the proof.
\end{proof}
Let $\Sigma \in \mathbb{R}^{p \times p}$ be symmetric positive definite. In the following, we work in the Hilbert space $(\mathbb{R}^p, \langle \cdot,\cdot \rangle_\Sigma)$, where for $x,y \in \mathbb{R}^p$,
\begin{equation*}
 \langle x,y \rangle_\Sigma := x^t \Sigma y,
\end{equation*}
and induced norm
\begin{equation*}
 \|x \|_\Sigma = \sqrt{ x^t \Sigma x}.
\end{equation*}
This means that projections and orthogonality etc. are always meant with respect to this space.
Let $\text{PA}$ denote the affine projection on the smallest affine space containing $b_1,\ldots,b_G$. Let $\text{PA}^{(g)}$ denote the affine projection on the smallest affine space containing $b_1,\ldots,b_{g-1},b_{g+1},\ldots,b_G$. Note that for $g=1$ this space can be expressed as $b_2 + \langle b_3-b_2,\ldots, b_G-b_2 \rangle$. Let $\Pi_B \in \mathbb{R}^{p \times p}$ denote the projection on $\langle b_2-b_1,\ldots, b_G-b_1 \rangle^\perp$.
\begin{lemma}\label{diff1}
 If $M_\Sigma(B)= \alpha_1 b_1+\ldots+\alpha_G b_G$ with $0 < \alpha_g < 1
 $ for $g =1,\ldots,G > 1$ and this representation is unique
 (i.e. $B=(b_1,...,b_G)$ has full rank), then $M_\Sigma$ is continuously differentiable in a neighborhood of $B$ with
\begin{align}\label{lemma}
 \Dif_{g,v} M_\Sigma(B) = &- \frac{\| M_\Sigma(B) \|_\Sigma}{\| (\text{Id}-\text{PA}^{(g)})b_g \|_\Sigma}    \langle \frac{M_\Sigma(B)}{\| M_\Sigma(B) \|_\Sigma},v \rangle_\Sigma \frac{(\text{Id}-\text{PA}^{(g)})b_g}{\| (\text{Id}-\text{PA}^{(g)})b_g \|_\Sigma} \nonumber \\
 &+\frac{\| (\text{Id}- \text{PA}^{(g)})M_\Sigma(B) \|_\Sigma}{\| (\text{Id}-\text{PA}^{(g)})b_g \|_\Sigma} \Pi_B v. 
\end{align}
Here, $ \Dif_{g,v} M_\Sigma(B)$ denotes the differential with respect to the variable $b_g$ in direction $v$.
\end{lemma}
\begin{remark}
In the proof of Theorem \ref{maintheorem}, we could assume that without loss of generality $\{1,\ldots,G \}=A(B,\Sigma)$, i.e.\ $B = B_{A(B,\Sigma)}$. We saw that in a neighborhood of $B$ and $\Sigma$, magging depends only on $B_{A(B,\Sigma)}$. Hence, for using the formula of $\Dif_g M_\Sigma(B)$ in the context of Theorem \ref{coverage} and \ref{maintheorem}, replace in the definition $B$ by $B_{A(B,\Sigma)}$. The derivatives with respect to $b_g$, $g \in \{1,\ldots,G \} - A(B,\Sigma)$ are zero.
\end{remark}
\begin{proof}
Without loss of generality, let us assume that $g = 1$. We will show that the partial derivatives exist and are continuous.

Let $ \Delta_1 \in \langle b_2-b_1,\ldots,b_G-b_1 \rangle^\perp$ and $\Delta_2 \in  \langle b_2-b_1,\ldots,b_G-b_1 \rangle$ and define $\tilde B := (b_1+\Delta_1+\Delta_2,b_2,\ldots,b_G)$. First, we want to show that, if $\| \Delta_1+\Delta_2 \|_\Sigma$ small,
\begin{equation}\label{almost}
M_\Sigma(\tilde B) = \text{PA}^{(1)}M_\Sigma(B)- \frac{\langle \text{PA}^{(1)}M_\Sigma(B) ,(\text{Id}-\text{PA}^{(1)})  \tilde b_1 \rangle_\Sigma}{\| (\text{Id}-\text{PA}^{(1)}) \tilde b_1\|_\Sigma^2} (\text{Id}-\text{PA}^{(1)}) \tilde b_1.
\end{equation}
Let us denote the r.h.s. by $\xi(\tilde B)$. We have to show:
\begin{enumerate}
 \item $\xi(\tilde B) \perp (\text{Id}-\text{PA}^{(1)}) \tilde b_1 $
 \item  $\xi(\tilde B) \perp \langle b_3-b_2,\ldots.,b_G-b_2 \rangle$
 \item $\xi(\tilde B) \in \text{CVX}(\tilde B)$, the convex hull generated
   by the columns of $\tilde B$.
\end{enumerate}
Note that 1. and 2. guarantee that the r.h.s. in \eqref{almost} is
perpendicular to the linear space generated by the columns of $\tilde B$.

1. is trivial. 2. By definition, $(\text{Id}-\text{PA}^{(1)}) \tilde b_1 \perp \langle b_3-b_2,\ldots.,b_G-b_2 \rangle$. $ \text{PA}^{(1)}M_\Sigma(B) \perp \langle b_3-b_2,\ldots.,b_G-b_2 \rangle$ as we can decompose into $\text{PA}^{(1)}M_\Sigma(B) = M_\Sigma(B) -(\text{Id}-\text{PA}^{(1)})M_\Sigma(B)$, which are both, by definition, perpendicular to $\langle b_3-b_2,\ldots.,b_G-b_2 \rangle$.

Now let us show 3.: $M_\Sigma(B) = \sum_{g=1}^G \alpha_g b_g $ for some $0 < \alpha_g $ and $\sum_{g=1}^G \alpha_g = 1$, i.e.\  $(B^t B)^{-1} B^t M_\Sigma(B) = (B_{ \restriction \langle b_1,\ldots,b_G \rangle})^{-1}M_\Sigma(B)=\alpha$.
Similarly, as $\xi(\tilde B )$ lies on the affine space generated by $\tilde b_1,\ldots,\tilde b_G$, we have $\xi(\tilde B ) = \sum_{g=1}^G \tilde \alpha_g \tilde b_g $ with $\sum_{g=1}^G \tilde \alpha_g = 1$. For small $\| \Delta_1+\Delta_2 \|_\Sigma$, $\tilde B$ has full rank and as $\xi(\tilde B) \rightarrow M_\Sigma(B)$,
\begin{equation*}
\lim_{\Delta \rightarrow 0} (\tilde B_{ \restriction \langle \tilde b_1,\ldots,\tilde b_G \rangle})^{-1} \xi(\tilde B) = \lim_{\Delta \rightarrow 0} (\tilde B^t \tilde B)^{-1} \tilde B^t \xi(\tilde B) = \alpha.
\end{equation*}
Hence, for small $\| \Delta_1+\Delta_2 \|_\Sigma$, $\tilde \alpha_g > 0$ and $\sum_{g=1}^G \tilde \alpha_g = 1$, hence $ \xi(\tilde B)  \in \text{CVX}(\tilde B)$ and thus $M_\Sigma(\tilde B)=\xi(\tilde B)$. This concludes the proof of \eqref{almost}.

Note that, as $\Delta_1 \perp \langle b_2-b_1,\ldots,b_G-b_1 \rangle= \langle b_1-b_2,b_3-b_2,\ldots,b_G-b_2 \rangle$,
\begin{align}\label{first}
 (\text{Id}-\text{PA}^{(1)} )  \tilde b_1 &= \tilde b_1 - \argmin_{\gamma \in b_2+\langle b_3-b_2,\ldots,b_G-b_2 \rangle} \|\gamma-b_1-\Delta_1-\Delta_2 \|_\Sigma^2 \nonumber \\
&= \tilde b_1 -  \argmin_{\gamma \in b_2+\langle b_3-b_2,\ldots,b_G-b_2 \rangle} \|\gamma-b_1-\Delta_2\|_\Sigma^2+ \| \Delta_1 \|_\Sigma^2 \nonumber \\
&= \Delta_1 + (\text{Id}- \text{PA}^{(1)} )(b_1 + \Delta_2).
\end{align}
$ (\text{Id}- \text{PA}^{(1)} )(b_1 +
\Delta_2)$ and $ (\text{Id}- \text{PA}^{(1)} )
b_1$ are linearly dependent. To see this, observe that both lie in the one-dimensional space $\langle b_2-b_1,\ldots,b_G-b_1 \rangle \cap  \langle b_3-b_2,\ldots,b_G-b_2 \rangle^\perp$. This implies that
\begin{align}\label{fourth}
 &\frac{\langle \text{PA}^{(1)}M_\Sigma(B) ,(\text{Id}-\text{PA}^{(1)})  (b_1+\Delta_2)\rangle_\Sigma}{\| (\text{Id}-\text{PA}^{(1)}) (b_1+\Delta_2)\|_\Sigma^2} (\text{Id}-\text{PA}^{(1)}) (b_1+\Delta_2) \nonumber \\
 &=  \frac{\langle \text{PA}^{(1)}M_\Sigma(B) ,(\text{Id}-\text{PA}^{(1)})  b_1\rangle_\Sigma}{\| (\text{Id}-\text{PA}^{(1)}) b_1\|_\Sigma^2} (\text{Id}-\text{PA}^{(1)}) b_1
\end{align}
Now we can put these pieces together: In the first step we use
\eqref{almost} and \eqref{first}, in the second we use $\Delta_1 \in
\langle b_2-b_1,\ldots,b_G-b_1 \rangle^\perp$. 
\begin{align*}
 &M_\Sigma(\tilde B)\\
=&\text{PA}^{(1)}M_\Sigma(B)- \frac{\langle \text{PA}^{(1)}M_\Sigma(B) ,\Delta_1 + (\text{Id}- \text{PA}^{(1)} )(b_1 + \Delta_2)\rangle_\Sigma}{\| \Delta_1 + (\text{Id}- \text{PA}^{(1)} )(b_1 + \Delta_2)\|_\Sigma^2} (\Delta_1 + (\text{Id}- \text{PA}^{(1)} )(b_1 + \Delta_2)) \\ 
=& \text{PA}^{(1)}M_\Sigma(B)- \frac{\langle \text{PA}^{(1)}M_\Sigma(B)
  ,\Delta_1 + (\text{Id}- \text{PA}^{(1)} )(b_1 +
  \Delta_2)\rangle_\Sigma}{\| \Delta_1 \|^2 + \|(\text{Id}- \text{PA}^{(1)}
  )(b_1 + \Delta_2)\|_\Sigma^2} (\Delta_1 + (\text{Id}- \text{PA}^{(1)}
)(b_1 + \Delta_2)).
\end{align*}
In the first step we do an expansion of the equation above and in the second, we use \eqref{fourth} and $ (\text{Id}- \text{PA}^{(1)} )(b_1 + \Delta_2) = (\text{Id}- \text{PA}^{(1)} )b_1 + \mathcal{O} (\| \Delta_2 \|_\Sigma) $:
\begin{align*}
 &M_\Sigma(\tilde B)\\
=&  \text{PA}^{(1)}M_\Sigma(B)- \frac{\langle \text{PA}^{(1)}M_\Sigma(B) , (\text{Id}- \text{PA}^{(1)} )(b_1 + \Delta_2)\rangle_\Sigma}{ \|(\text{Id}- \text{PA}^{(1)} )(b_1 + \Delta_2)\|_\Sigma^2}  (\text{Id}- \text{PA}^{(1)} )(b_1 + \Delta_2)\\
&- \frac{\langle \text{PA}^{(1)}M_\Sigma(B) ,\Delta_1\rangle_\Sigma}{ \|(\text{Id}- \text{PA}^{(1)} )(b_1 + \Delta_2)\|_\Sigma^2} (\text{Id}- \text{PA}^{(1)} )(b_1 + \Delta_2) \\
&- \frac{\langle \text{PA}^{(1)}M_\Sigma(B) , (\text{Id}- \text{PA}^{(1)} )(b_1 + \Delta_2)\rangle_\Sigma}{\|(\text{Id}- \text{PA}^{(1)} )(b_1 + \Delta_2)\|_\Sigma^2} \Delta_1  + \mathcal{O}(\| \Delta_1\|_\Sigma^2 +\| \Delta_2 \|_\Sigma^2)) \\
=&  \text{PA}^{(1)}M_\Sigma(B)- \frac{\langle \text{PA}^{(1)}M_\Sigma(B) , (\text{Id}- \text{PA}^{(1)} )b_1 \rangle_\Sigma}{ \|(\text{Id}- \text{PA}^{(1)} )b_1 \|_\Sigma^2}  (\text{Id}- \text{PA}^{(1)} )b_1 \\
&- \frac{\langle \text{PA}^{(1)}M_\Sigma(B) ,\Delta_1\rangle_\Sigma}{ \|(\text{Id}- \text{PA}^{(1)} )b_1\|_\Sigma^2} (\text{Id}- \text{PA}^{(1)} )b_1 \\
&- \frac{\langle \text{PA}^{(1)}M_\Sigma(B) , (\text{Id}- \text{PA}^{(1)} )b_1\rangle_\Sigma}{\|(\text{Id}- \text{PA}^{(1)} )b_1\|_\Sigma^2} \Delta_1  + \mathcal{O}(\| \Delta_1\|_\Sigma^2 +\| \Delta_2 \|_\Sigma^2)).
\end{align*}
From this and \eqref{almost} we obtain
\begin{align*}
 &M_\Sigma(\tilde B) - M_\Sigma( B) \\
 =& - \frac{\langle \text{PA}^{(1)}M_\Sigma(B) ,\Delta_1\rangle_\Sigma}{ \|(\text{Id}- \text{PA}^{(1)} )b_1\|_\Sigma^2} (\text{Id}- \text{PA}^{(1)} )b_1 \\
&- \frac{\langle \text{PA}^{(1)}M_\Sigma(B) , (\text{Id}- \text{PA}^{(1)} )b_1\rangle_\Sigma}{\|(\text{Id}- \text{PA}^{(1)} )b_1\|_\Sigma^2} \Delta_1  + \mathcal{O}(\| \Delta_1\|_\Sigma^2 +\| \Delta_2 \|_\Sigma^2)).
\end{align*}
Now let us write $\Delta_1+\Delta_2 = \gamma v, \Delta_1 = \gamma  ( M_\Sigma(B) / \| M_\Sigma(B)  \|_\Sigma + \mu v_\perp)$ with $ v_\perp \perp M_\Sigma(B) $ and $v_\perp \perp \langle b_2 - b_1,\ldots,b_G-b_1 \rangle$. By noting that 
\begin{align*}
\langle \text{PA}^{(1)}M_\Sigma(B) ,  \Delta_1\rangle_\Sigma &= \langle M_\Sigma(B)+ (\text{PA}^{(1)}-\text{Id})M_\Sigma(B) , \gamma \frac{  M_\Sigma(B) }{ \| M_\Sigma(B)  \|_\Sigma} + \mu v_\perp)  \rangle_\Sigma \\
&= \gamma \| M_\Sigma(B)\|_\Sigma \\
&= \gamma \langle \frac{  M_\Sigma(B) }{ \| M_\Sigma(B)  \|_\Sigma},v \rangle_\Sigma \| M_\Sigma(B)\|_\Sigma,
\end{align*}
and, as $ (\text{Id}- \text{PA}^{(1)})M_\Sigma(B)$ and $
(\text{Id}-\text{PA}^{(1)})  b_1$ are linearly dependent (both lie in the one-dimensional space $\langle b_2-b_1,\ldots,b_G-b_1 \rangle \cap  \langle b_3-b_2,\ldots,b_G-b_2 \rangle^\perp$),
\begin{align*}
-\langle \text{PA}^{(1)}M_\Sigma(B) ,(\text{Id}-\text{PA}^{(1)})  b_1\rangle_\Sigma
&= \langle (\text{Id}- \text{PA}^{(1)})M_\Sigma(B) ,(\text{Id}-\text{PA}^{(1)})  b_1\rangle_\Sigma \\
&= \|( \text{Id}- \text{PA}^{(1)})M_\Sigma(B) \|_\Sigma \|(\text{Id}-\text{PA}^{(1)})  b_1\|_\Sigma.
\end{align*}
We obtain:
\begin{align*}
 &M_\Sigma(\tilde B) - M_\Sigma( B) \\
 =& - \gamma \frac{ \| M_\Sigma(B)\|_\Sigma }{ \|(\text{Id}- \text{PA}^{(1)} )b_1\|_\Sigma^2}  \langle \frac{  M_\Sigma(B) }{ \| M_\Sigma(B)  \|_\Sigma},v \rangle_\Sigma (\text{Id}- \text{PA}^{(1)} )b_1 \\
&- \gamma \frac{\|( \text{Id}- \text{PA}^{(1)})M_\Sigma(B) \|_\Sigma \|(\text{Id}-\text{PA}^{(1)})  b_1\|_\Sigma}{\|(\text{Id}- \text{PA}^{(1)} )b_1\|_\Sigma^2} \Pi_B v  + \mathcal{O}(\| \Delta_1\|_\Sigma^2 +\| \Delta_2 \|_\Sigma^2)).
\end{align*}
Hence the directional derivative exists and is equal to \eqref{lemma}. The assertion follows by existence and continuity of the directional derivatives in a neighborhood of $B$.
\end{proof}

\begin{lemma}\label{continuity}
Let $\Sigma^0$ be positive definite. $M_\Sigma(B)$ is continuous in $B$ and
$\Sigma$ in a neighborhood of $\Sigma^0$. Furthermore, under the assumptions of Theorem \ref{coverage} (or Theorem \ref{maintheorem}), in a neighborhood of $B^0$ and $\Sigma^0$, $A(B,\Sigma)$ is constant. 
\end{lemma}

\begin{proof}
First, let us prove that magging is continuous. Proof by contradiction:
Assume there exist sequences $B_k \rightarrow B$, $\Sigma_k \rightarrow
\Sigma$ positive definite such that $M_{\Sigma_k}(B_k) \not \rightarrow
M_{\Sigma}(B)$. Without loss of generality, as $\Sigma$ is invertible, $M_{\Sigma_k}(B_k)$ converges, too. By definition of $M_{\Sigma_k}(B_k)$ we have
\begin{equation*}
 \| M_{\Sigma_k}(B_k) \|_{\Sigma_k} \le \| \Pi_{B_k} M_{\Sigma}(B) \|_{\Sigma_k},
\end{equation*}
where $\Pi_{B_k}$ denotes the projection (in $\langle \cdot , \cdot \rangle$) on the convex set $\text{CVX}(B_k)$. By continuity,
\begin{equation*}
 \| \lim_k M_{\Sigma_k}(B_k) \|_{\Sigma} \le \| M_{\Sigma}(B) \|_{\Sigma}.
\end{equation*}
We have $M_{\Sigma_k}(B_k) \in \text{CVX}(B_k)$ and hence by continuity $\lim_k M_{\Sigma_k}(B_k) \in \text{CVX}(B)$. As magging is unique ($\Sigma$ is positive definite), this yields a contradiction.

Consider $b_g^0$ with $g \in A(B^0,\Sigma^0) $. By the assumptions of
Theorem \ref{coverage}, $M_{\Sigma^0}(B^0) = \sum_{i \in A(B^0,\Sigma^0)}
\alpha_i b_i^0$ with $0 < \alpha_i < 1$. Hence for small $\gamma \in \mathbb{R}$, $(1-\gamma) M_{\Sigma^0}(B^0) + \gamma b_g^0 \in \text{CVX}(B^0)$ and by definition of magging
\begin{equation}\label{perpeq}
\| M_{\Sigma^0}(B^0) \|_{\Sigma^0} \le \| (1-\gamma) M_{\Sigma^0}(B^0) + \gamma b_g^0  \|_{\Sigma^0}
\end{equation}
Using this inequality for small $\gamma >0$ and small $\gamma < 0$ we obtain $\langle M_{\Sigma^0}(B^0) , b_g^0- M_{\Sigma^0}(B^0) \rangle =0$. Hence, for all $g \in A(B^0,\Sigma^0)$, $M_{\Sigma^0}(B^0)$ is perpendicular (with respect to $\langle \cdot , \cdot\rangle_{\Sigma^0}$) to $b_g^0- M_{\Sigma^0}(B^0)$. Hence $A(B^0,\Sigma^0) \subset M_{\Sigma^0}(B^0)+M_{\Sigma^0}(B^0)^\perp$.

Furthermore, by assumptions of Theorem \ref{coverage}, if $g \not \in
A(B^0,\Sigma^0) $ we have $b_g^0 \not \in M_{\Sigma^0}(B^0)+M_{\Sigma^0}(B^0)^\perp$. By
continuity, for $ B=(b_1,...,b_G)$ close to $B^0$ and $ \Sigma$ close to
$\Sigma^0$ (in $\|\cdot \|_2) $ we have $ b_g \not \in M_{\Sigma}( B)+M_{ \Sigma}( B)^\perp$. By an analogous
argument as in equation~\eqref{perpeq}, $g \not \in  A( B,\Sigma)$. This proves $ A(B^0,\Sigma^0)  \subset A( B, \Sigma) $.

It remains to show  $   A( B, \Sigma) \subset A(B^0,\Sigma^0) $: For notational simplicity let us assume $A(B^0,\Sigma^0) = \{1,\ldots,G \}$. For  $ B$ close to $B^0$ and $ \Sigma$ close to $\Sigma^0$, $M_{ \Sigma}( B) =  B  \tilde \alpha$ with $\sum_{i=1}^G \tilde \alpha_i = 1$, $0 \le \tilde \alpha_i \le 1$. We want to show that for $ B$ close to $B^0$ and $ \Sigma$ close to $\Sigma^0$ (in $\|\cdot \|_2) $, $0 < \tilde \alpha_i < 1$.

To this end, note that by the assumptions of Theorem \ref{coverage} we have that $B_{A(B^0,\Sigma^0)}^0$ (here without loss of generality: $B^0$) has full rank, hence for $ B$ close to $B^0$ and $\Sigma$ close to $\Sigma^0$, $\left( B^t  B\right)^{-1}  B^t M_{ \Sigma}( B) = \tilde \alpha$ with $\tilde \alpha_i \ge 0$, $\sum_i \tilde \alpha_i =1$. Furthermore,
\begin{equation*}
\lim_{ B \rightarrow B^0, \Sigma \rightarrow \Sigma^0} \left( B^t  B\right)^{-1}  B^t M_{ \Sigma}( B) = \left( (B^0)^t  B^0\right)^{-1}  (B^0)^t M_{ \Sigma^0}( B^0) = \alpha.
\end{equation*}
Hence for $ B$ close to $B^0$ and $ \Sigma$ close to $\Sigma^0$ (in $\|\cdot \|_2) $), $0 < \tilde \alpha_i < 1$. This concludes the proof.
\end{proof}

\begin{lemma}\label{diff2}
Let $G >2$. Let $M_\Sigma(B)= \alpha_1 b_1+\ldots+\alpha_G b_G$ with unique $0<\alpha_g <
1$ satisfying $\sum_{g=1}^G \alpha_g =1$. Then the mapping
\begin{align*}
\{ \text{positive definite matrices in $\mathbb{R}^{p \times p}$} \} &\rightarrow \mathbb{R}^p \\
\Sigma &\mapsto M_\Sigma(B)
\end{align*}
is continuously differentiable at $B$, $\Sigma$. Let $\Delta$ be a symmetric matrix. The differential in direction $\Delta$ is
\begin{equation*}
\mathrm{D}_\Sigma M_\Sigma(B)\Delta = -D ( D^t \Sigma D )^{-1} D^t \Delta M_\Sigma(B),
\end{equation*}
where
\begin{equation*}
D := (b_2,\ldots,b_G) - (b_1,\ldots,b_1).
\end{equation*}

\end{lemma}

\begin{proof}
By elementary analysis, it suffices to show that the directional derivatives exist in a neighborhood and that they are continuous.

For a small symmetric pertubation  $\lambda \Delta$, by continuity of magging (Lemma~\ref{continuity}), $M_{\Sigma+\lambda\Delta}(B)$ has to satisfy
\begin{equation*}
 M_{\Sigma+\lambda\Delta}(B) = M_{\Sigma}(B)+D \gamma
\end{equation*}
for some (small) vector $\gamma  \in \mathbb{R}^{G-1}$. By definition of magging, and as $0<\alpha_g<1$ we have  $
\|M_{\Sigma+\lambda\Delta}(B) \|_{\Sigma+\lambda\Delta} \le \|
M_{\Sigma+\lambda\Delta}(B)+D \gamma'  \|_{\Sigma+\lambda\Delta}$ for all
small vectors $\gamma' \in \mathbb{R}^{G-1}$. Hence,
\begin{equation}\label{diffcond}
M_{\Sigma+\lambda\Delta}(B)^t (\Sigma+\lambda\Delta) D = 0.
\end{equation}
Putting these two conditions together, we get
\begin{equation*}
 (M_{\Sigma}(B)+D \gamma)^t (\Sigma+\lambda\Delta) D = 0.
\end{equation*}
Furthermore, analogously as in equation~\eqref{diffcond} we obtain
\begin{equation*}
M_{\Sigma}(B)^t \Sigma D = 0.
\end{equation*}
By combining the last two equations,
\begin{equation*}
\gamma^t D^t (\Sigma + \lambda \Delta)D = -M_{\Sigma}(B)^t \lambda\Delta D.
\end{equation*}
As $D^t (\Sigma + \lambda \Delta)D$ is invertible ($D$ has full rank as $B$
has full rank. $B$ has full rank as the $\alpha_g$ are unique),
\begin{align*}
\gamma^t &= -M_{\Sigma}(B)^t \lambda\Delta D (D^t (\Sigma + \lambda\Delta)D)^{-1},\\
D\gamma &= -  D(D^t (\Sigma + \lambda\Delta)D)^{-1} D^t \lambda \Delta M_{\Sigma}(B).
\end{align*}
Dividing by $\lambda$ and letting $\lambda \rightarrow 0$ gives the desired result.
\end{proof}
\begin{lemma}\label{clt}
Let $\bold{X}_{k \cdot}\sim F$, $k = 1,...,nG$ denote the i.i.d.\ rows of
$\bold{X}$. Let $\mathbb{E}[\| \bold{X}_{1 \cdot}^t \bold{X}_{1 \cdot}
\|_2^2] < \infty$ and $\Sigma^0 = \mathbb{E} [ \bold{X}_{1 \cdot}^t
\bold{X}_{1 \cdot}] $ positive definite. Then, for $n \rightarrow \infty$,
\begin{equation*}
\frac{1}{G\sqrt{n}}\sum_{k=1}^{nG} \left( \bold{X}_{k \cdot}^t \bold{X}_{k \cdot}-\Sigma^0 \right) \rightharpoonup \Delta
\end{equation*}
where the symmetric matrix $\Delta$ has centered multivariate normal distributed entries under and on the diagonal with covariance 
\begin{equation*}
c_{ijkl} := \text{Covar}(\Delta_{ij},\Delta_{kl}) = \frac{1}{G}\mathbb{E}[ ( \bold{X}_{1 i}  \bold{X}_{1 j}-\mathbb{E}[ \bold{X}_{1 i}  \bold{X}_{1 j}] )(\bold{X}_{1 k}  \bold{X}_{1 l}-\mathbb{E}[\bold{X}_{1 k}  \bold{X}_{1 l}])].
\end{equation*}
\begin{proof}
Apply the CLT.
\end{proof}
\end{lemma}
In the following Lemma, we want to calculate the distribution of
\begin{equation*}
 -D ( D^t \Sigma D )^{-1} D^t \Delta M_\Sigma(B).
\end{equation*}

\begin{lemma}\label{matcalc}
Let us use setting of Lemma \ref{diff2} and \ref{clt}.
\begin{equation*}
\Dif_\Sigma M_\Sigma(B) \sqrt{n} (\hat \Sigma - \Sigma) \rightharpoonup \mathcal{N}(0,V(B,\Sigma))
\end{equation*}
with
\begin{equation*}
V(B,\Sigma) =  D ( D^t \Sigma D )^{-1} D^t C D ( D^t \Sigma D )^{-1} D^t,
\end{equation*}
where 
\begin{equation*}
C_{ij} = \sum_{k,l=1}^p M_\Sigma(B)_k M_\Sigma(B)_l c_{iklj},
\end{equation*}
is the covariance matrix of $ \Delta M_{\Sigma}(B)$ and
\begin{equation*}
D := (b_2,\ldots,b_G) - (b_1,\ldots,b_1).
\end{equation*}

\end{lemma}

\begin{remark}
In the proof of Theorem \ref{maintheorem}, we could assume that without loss of generality $\{1,\ldots,G \}=A(B,\Sigma)$, i.e.\ $B = B_{A(B,\Sigma)}$. For using the definition of $V$ in the context of Theorem \ref{coverage} and \ref{maintheorem}, replace in the definition $B$ by $B_{A(B,\Sigma)}$. The $G$ in the definition of $C$ stays the same, i.e.\ it is still the total number of groups.
\end{remark}

\begin{proof}

With Lemma \ref{diff2} and \ref{clt} it suffices to calculate the distribution of 
\begin{equation*}
 -D ( D^t \Sigma D )^{-1} D^t \Delta M_\Sigma(B),
\end{equation*}
i.e.\ the nontrivial part is to calculate the distribution of $\Delta M_\Sigma(B)$. We know it is Gaussian and centered, hence it suffices to determine the covariance matrix:
\begin{align*}
\mathbb{E} \left( \Delta M_\Sigma(B) M_\Sigma(B)^t  \Delta \right)_{ij} &= \mathbb{E} \sum_{k,l=1}^p \Delta_{ik} (M_\Sigma(B) M_\Sigma(B)^t)_{kl}  \Delta_{lj} \\
&= \sum_{k,l=1}^p M_\Sigma(B)_k M_\Sigma(B)_l \mathbb{E}  \Delta_{ik}  \Delta_{lj} \\
&=\sum_{k,l=1}^p M_\Sigma(B)_k M_\Sigma(B)_l c_{iklj}.
\end{align*}
In the last line we used Lemma \ref{clt}. This concludes the proof.
\end{proof}

\subsection{Proof of Theorem \ref{coverage}}

\begin{proof}
First, note that by Lemma \ref{diff1}, $W(\Sigma^0,B^0)$ is invertible. Using Lemma \ref{continuity}, in a neighborhood of $B^0$ and $\Sigma^0$ the set-valued function $A(B,\Sigma)$ is constant. Hence, by Lemma \ref{diff1} and Lemma \ref{diff2}, the derivatives of $M_{\Sigma}(B)=M_{\Sigma}(B_{A(B,\Sigma)})$ are continuous at $B^0$ and $\Sigma^0$. Furthermore, $V(B_{A(B,\Sigma)},\Sigma)$ is continuous in $C$ and in $B$ and $\Sigma$ at $B^0$ and $\Sigma^0$. All together, $W(\Sigma,B)$ is continuous at $B^0$ and $\Sigma^0$ in all its variables. By the definition of $C$ in Lemma~\ref{matcalc} and the definition of $\hat C$ in Section~\ref{defs}, $\hat C \rightarrow C$.

Hence, $W(\hat \Sigma,\hat B) \rightarrow W( \Sigma^0,B^0)$ in probability and we obtain that $W(\hat B,\hat \Sigma)^{-1} \rightarrow W(B^0,\Sigma^0)^{-1}$ in probability. By Theorem \ref{maintheorem} and Slutsky's Theorem we obtain
\begin{equation*}
\sqrt{n}(M_{\hat \Sigma}(\hat B)-M_{\Sigma^0}(B^0))^t W(\hat B,\hat \Sigma)^{-1} \sqrt{n} (M_{\hat \Sigma}(\hat B)-M_{\Sigma^0}(B^0)) \rightharpoonup \chi^2(p)
\end{equation*}
for $n \rightarrow \infty$. Hence
\begin{align*}
&\mathbb{P} [ M_{ \Sigma^0}( B^0) \in \bold{C}(\hat \Sigma,\hat B)] \\
=&\mathbb{P} [ (M_{\hat \Sigma}(\hat B)-M_{\Sigma^0}(B^0))^t W(\hat B,\hat \Sigma)^{-1}  (M_{\hat \Sigma}(\hat B)-M_{\Sigma^0}(B^0)) \le \frac{\tau}{n}]\\  \rightarrow &1-\alpha
\end{align*}
for $n \rightarrow \infty$. This concludes the proof.

\end{proof}

\subsection{Relaxation-based approach}\label{relaxation}

A simple approach is as follows: For given $\alpha > 0$, take random sets $\mathcal{R}_B$, $\mathcal{R}_\Sigma$ such that

\begin{equation*}
\mathbb{P}[\Sigma^0 \in \mathcal{R}_\Sigma, B^0  \in \mathcal{R}_B] \ge 1-\alpha,
\end{equation*}
where $B^0 = (b_1^0,\ldots,b_G^0)$ is the matrix of regression coefficients in all $G$ groups. A generic approach is to choose a confidence region for $\Sigma^0$ on the confidence level $1-\alpha/2$ and confidence regions for $b_g^0$ on the confidence level $1-\alpha/(2G)$. However, this approach can easily be improved by taking larger regions around $\hat b_g$ that are far away from zero (thus have negligible influence on $M_{\hat \Sigma}(\hat B)$)  and smaller regions around $\hat b_g$ that are close to zero. Then calculate
\begin{equation*}
\mathcal{R} = \{ M_{\tilde \Sigma}(\tilde B) : \tilde \Sigma \in \mathcal{R}_\Sigma, \tilde B \in \mathcal{R}_B
\} \subset \mathbb{R}^p,
\end{equation*}
which is a $1-\alpha$ confidence region for the maximin effect. However, direct computation of this confidence region is computationally cumbersome.

For known $\Sigma^0$ the idea can be relaxed to the following scheme:

For $m \in \mathbb{R}^p$ and $\Sigma \in \mathbb{R}^{p \times p}$ positive definite let us define $\|m\|_{\Sigma} := \sqrt{m^T\Sigma
  m}$. Note that this defines a norm on $\mathbb{R}^p$. Now, 
\begin{align*} \|M_{\Sigma^0}(B')\|_{\Sigma^0} 
=&\min_{\gamma \ge 0, \sum_{g=1}^G \gamma_g =1} \|B' \gamma \|_{\Sigma^0} \\
=& \min_{\gamma \ge 0, \sum_{g=1}^G \gamma_g =1} \|B' \gamma \|_{\Sigma^0}-\|B \gamma\|_{\Sigma^0}+\|B \gamma\|_{\Sigma^0} \\
\le& \sup_{\gamma \ge 0, \sum_{g=1}^G \gamma_g =1} | \|B' \gamma \|_{\Sigma^0}-\|B \gamma \|_{\Sigma^0} | + \min_{\gamma \ge 0, \sum_{g=1}^G \gamma_g =1} \|B \gamma\|_{\Sigma^0} \\
\le &\sup_{\gamma \ge 0, \sum_{g=1}^G \gamma_g =1}  \|(B'-B) \gamma \|_{\Sigma^0} + \min_{\gamma \ge 0, \sum_{g=1}^G \gamma_g =1} \|B \gamma\|_{\Sigma^0}\end{align*}
and hence
\begin{align*}
\|M_{\Sigma^0}(B')\|_{\Sigma^0}  \le & \sup_{\gamma \ge 0, \sum_{g=1}^G \gamma_g =1} \sum_{g=1}^G \gamma_g \| b'_g-b_g\|_{\Sigma^0} + \min_{\gamma \ge 0, \sum_{g=1}^G \gamma_g =1} \|B \gamma\|_{\Sigma^0}\\ 
= &  \max_{g=1,\ldots,G} \| b'_g-b_g \|_{\Sigma^0} + \min_{\gamma \ge 0, \sum_{g=1}^G \gamma_g =1} \|B \gamma\|_{\Sigma^0} \\
= &  \max_{g=1,\ldots,G} \| b'_g-b_g \|_{\Sigma^0} +  \|M_\Sigma(B)\|_{\Sigma^0}
\end{align*}
By symmetry,
\begin{align}\label{genapproach}
|\|M_{\Sigma^0}(B')\|_{\Sigma^0}-\|M_{\Sigma^0}(B)\|_{\Sigma^0} | \le  \max_{g=1,\ldots,G} \| b'_g-b_g \|_{\Sigma^0}.
\end{align}
We can now choose a covering of the confidence region $\mathcal{R}_B$ with $B^{(k)} \in \mathcal{R}_B$,$k=1,\ldots,K$ such that balls
$\mathcal{B}_{\epsilon_k}(B^{(k)})$ with radius $\epsilon_k$ around $B^{(k)}$ cover
 $\mathcal{R}_B$ with respect to the maximum norm $\|B \|_{\mbox{max}} :=
\max_g \|b_g \|_{\Sigma^0}$. 

A confidence region of the maximin effect can then be constructed as
\begin{equation*}
\tilde{\mathcal{R}} = \bigcup_{k=1,\ldots,K}  \{ M : | \| M \|_{\Sigma^0} -\|
M_{\Sigma^0}(B^{(k)}) \|_{\Sigma^0} | \le \epsilon_k \} \cap \text{CVX} \left(  \mathcal{B}_{\epsilon_k}(B^{(k)}) \right).
\end{equation*}
This confidence region is valid: For all $M_{\Sigma^0}(B') \in \mathcal{R}_B$ there
exists $k \in \{ 1,\ldots,K\}$ such that $\| B'-B^{(k)}
\|_{\text{max}} \le \epsilon_k$. By equation \eqref{genapproach},
$|\|M_{\Sigma^0}(B')\|_{\Sigma^0}-\|M_{\Sigma^0}(B^{(k)})\|_{\Sigma^0} | \le
\epsilon_k$, hence $M_{\Sigma^0}(B') \in \tilde{\mathcal{R}}_{B}$. This implies  $\mathcal{R}_B
\subset \tilde{\mathcal{R}}_B $;
\begin{equation*}
\mathbb{P}[M_{\Sigma^0}(B^0)  \in \tilde{\mathcal{R}}] \ge
\mathbb{P}[M_{\Sigma^0}( B^0)  \in \mathcal{R}] \ge \mathbb{P}[ B^0  \in \mathcal{R}_B] \ge 1-\alpha.
\end{equation*}
If $\Sigma^0$ is unknown, using the approach above we need to estimate lower and upper bounds for $\| \cdot \|_{\Sigma^0}$.
\end{document}